\documentclass[envcountsame]{llncs}
    \usepackage{amsmath}
    \usepackage{amsxtra}
    \usepackage{amstext}
    \usepackage{amssymb}
    \usepackage{latexsym}
    \usepackage{dsfont} 
\newcommand{\GF}[1]{{\mathbb F}_{#1}}

\begin{document}
\title{Solutions of $x^{q^k}+\cdots+x^{q}+x=a$ in $\GF{2^n}$}
\author{Kwang Ho Kim\inst{1,2} \and  Jong Hyok Choe\inst{1} \and Dok Nam Lee\inst{1}  \and Dae Song Go\inst{3} \and  Sihem Mesnager\inst{4} } \institute{ Institute of Mathematics,
State Academy of Sciences, Pyongyang, Democratic People's Republic
of Korea \email{khk.cryptech@gmail.com} \and PGItech Corp.,
Pyongyang, Democratic People's Republic of Korea \and Master School,
University of Natural Science, Pyongyang, Democratic People's
Republic of Korea\and LAGA, Department of Mathematics, University of
Paris VIII and Paris XIII, CNRS and Telecom ParisTech, France
\email{smesnager@univ-paris8.fr}}%

\maketitle

\begin{abstract}
Though it is well known that the roots of any affine polynomial over
a finite field can be computed by a system of linear equations by
using a normal base of the field, such solving approach appears to
be difficult to apply when the field is fairly large. Thus, it may
be of great interest to find an explicit representation of the
solutions independently of the field base. This was previously done
only for quadratic equations over binary finite field. This paper
gives an explicit representation of solutions for a much wider class of
affine polynomials over a binary prime field.

\end{abstract}

\section{Introduction}

Define
\[T_l^k(x):=x+x^{2^l}+\cdots+x^{2^{l(k/l-2)}}+x^{2^{l(k/l-1)}}\] when
$l|k$, and in particular
\[T_k(x):=T_1^k(x)=x+x^2+\cdots+x^{2^{k-2}}+x^{2^{k-1}}.\] The degree
of the polynomial $T_l^k(x)$ is $2^{k-l}$.

This paper gives the explicit representations of all
$\overline{\GF{2}}-$ and $\GF{2^n}-$solutions to the affine equation
\[ T_l^{k}(x)=a, a\in \GF{2^n}.\] Obviously, this equation has no
multiple roots since $(T_l^k)'=1\neq 0$. Throughout this paper, we
set $d=\gcd(n,k)$.

To the best of our knowledge, following is the only previous result
in this direction.
\begin{lemma}\label{quadratic_old}(Page 26 of \cite{BSS1999}, 11.1.120 of \cite{MP2013})
The quadratic equation
\[
x^2+x+a=0, a\in \GF{2^n}
\]
has solutions in $\GF{2^n}$ if and only if
\[
T_n(a)=0.
\]
Let us assume $T_n(a)=0$. Let $\delta$ be an element in $\GF{2^n}$
such that $T_n(\delta)=1$ (if $n$ is odd, then one can take
$\delta=1$). Then,
\[x_0=\sum_{i=0}^{n-2}(\sum_{j=i+1}^{n-1}\delta^{2^j})a^{2^i}\] is a
solution to the equation.
\end{lemma}

\section{Some useful facts}

\begin{lemma}\label{lem_properties} For any positive integers $k, k', l, l'$ such that $s|l|k$ and $l'|k'$, followings
hold.
\begin{enumerate}
\item (Commutativity)
\[
T_l^k\circ T_{l'}^{k'}=T_{l'}^{k'}\circ T_l^k.
\]
\item (Transitivity)
\[
T_l^k\circ T_s^l=T_s^k.
\]
\item
\[
T_k\circ T_2(x)=T_k^{2k}(x)=x+x^{2^k}.
\]
\item
\[
T_k\circ T_k\circ T_2 =T_{2k}.
\]
\end{enumerate}
\end{lemma}
\begin{proof}
All statements can be easily checked by direct calculation. \qed
\end{proof}

\begin{lemma}\label{lem_fielddef} For any positive integers $n$ and $k$, it holds
\[
T_k(x)\in \GF{2^n}\Longleftrightarrow T_n(x)\in \GF{2^k}.
\]
In particular, letting $k=1$, we have
\[
x\in \GF{2^n} \Longleftrightarrow T_n(x)\in \GF{2}
\Longleftrightarrow T_n\circ T_2(x)=0.
\]
\end{lemma}
\begin{proof}
Since $T_k(x)+T_k(x)^{2^n}=T_n(x)+T_n(x)^{2^k}$ which is checked by
direct computation, it follows $T_k(x)\in
\GF{2^n}\Longleftrightarrow T_k(x)^{2^n}=T_k(x) \Longleftrightarrow
T_k(x)+T_k(x)^{2^n}=0 \Longleftrightarrow
T_n(x)+T_n(x)^{2^k}=0\Longleftrightarrow T_n(x)\in \GF{2^k}$. \qed

\end{proof}

Following fact, though already well-known, can be reformulated.
\begin{corollary}\label{T_image} Let $l$ be a divisor of $k$. Then,
\[
T_l^k(\GF{2^k})=\GF{2^l}.
\]
\end{corollary}
\begin{proof}
This follows from the fact that $T_l(
T_l^k(\GF{2^k}))=T_k(\GF{2^k})=\GF{2}$. \qed
\end{proof}

\begin{theorem}\label{gcd_1} Let us assume $\gcd(n,k)=1$. Then it holds
\[
T_k(x)\in \GF{2^n}\Longleftrightarrow x\in \GF{2^n}+\GF{2^k},
\]
where $\GF{2^n}+\GF{2^k}=\{a+b\,|\,a\in\GF{2^n}, b\in\GF{2^k}\}$.
\end{theorem}

\begin{proof}
$(\Longleftarrow)$\\
Let $x=a+b$ for $a\in\GF{2^n}$ and $b\in\GF{2^k}$. Then $a^{2^n}=a$,
and $T_k(b)\in \GF{2}$ by above proposition, thus
\begin{align*}
T_k(a+b)^{2^n}&=(T_k(a)+T_k(b))^{2^n}\\
&=T_k(a)^{2^n}+T_k(b)^{2^n}\\
&=T_k(a^{2^n})+T_k(b)\\
&=T_k(a)+T_k(b)=T_k(a+b),
\end{align*}
where the linearity of $T_k$ was exploited. That is
$T_k(x)=T_k(a+b)\in\GF{2^n}$.

$(\Longrightarrow)$\\
Since the necessity in the statement has been proved, in order to
prove the sufficiency in the statement, it is enough to show
\[
\#\{x\in \overline{\GF{2}}\,|\,T_k(x)\in
\GF{2^n}\}=\#\{\GF{2^n}+\GF{2^k}\},
\]
where $\overline{\GF{2}}$ is the algebraic closure of $\GF{2}$.
 To begin with, we have $\#\{x\in \overline{\GF{2}}\,|\,T_k(x)\in
\GF{2^n}\}=2^{n+k-1}$ because for every $a\in \GF{2^n}$ the equation
$T_k(x)=a$ has $2^{k-1}$ different solutions.

 On the other hand, it also holds
 $\#\{\GF{2^n}+\GF{2^k}\}=2^{n+k-1}$. In fact, for $a, a'\in
 \GF{2^n}$ and $b, b'\in \GF{2^k}$, it holds $a+b=a'+b'$ $\iff$ $a+a'=b+b'\in
 \GF{2^n}\cap\GF{2^k}=\GF{2}$, i.e., ($a=a'$ and $b=b'$) or  ($a=a'+1$ and
 $b=b'+1$). Therefore $\#\{\GF{2^n}+\GF{2^k}\}=\#\{a+b\,|\,a\in\GF{2^n}, b\in\GF{2^k}\}=(2^n\cdot2^k)/2=2^{n+k-1}$.
 \qed
\end{proof}

Without the condition $\gcd(n,k)=1$, we give:
\begin{theorem} It holds
\[
\{x\in \overline{\GF{2}}\,|\,T_k(x)\in \GF{2^n}\}=\{T_n^{[n,k]}\circ
T_k^{[n,k]}\circ T_2(x) \,|\,x\in \GF{2^{2[n,k]}}\}(\subset
\GF{2^{2[n,k]}}),
\]
where $[n,k]$ is the least common multiple  of two integers $n$ and
$k$.
\end{theorem}
\begin{proof} Let $L=[n,k]$.
Let us set $y=T_n^{L}\circ T_k^{L}\circ T_2(x)$ for $x\in
\GF{2^{2L}}$. To begin with, we will show $T_n\circ T_k(y)\in
\GF{2}$ which is equivalent to $T_k(y)\in \GF{2^n}$ by Lemma
\ref{lem_fielddef}. In fact,
\[
T_n\circ T_k(y)=T_n^{L}\circ T_n\circ T_k^{L}\circ T_k\circ
T_2(x)=T_L\circ T_L \circ T_2(x)=T_{2L}(x)\in \GF{2},
\]
where the equalities are from Lemma \ref{lem_fielddef}, except for
the last equality which is from Lemma \ref{lem_properties}.

Obviously, the cardinality of the left side set is $2^{n+k-1}$. On
the other hand, the cardinality of the right side set is also
$2^{n+k-1}$ as it equals $2^{2L-(L-n)-(L-k)-1}=2^{n+k-1}$ and so the
two sets coincide. \qed
\end{proof}

\begin{example}
By Theorem \ref{gcd_1}, we know that $\{x\in
\overline{\GF{2}}\,|\,T_k(x)\in \GF{2^n}\}\subset \GF{2^{nk}}=
\GF{2^{[n,k]}}$ when $\gcd(n,k)=1$. However, it is not always the
case. Let us consider the case $n=k=2$.
\begin{align*}
&\{x\in \overline{\GF{2}}\,|\,T_2(x)\in \GF{2^2}\}=\{x\in
\overline{\GF{2}}\,|\,x+x^2+x^4+x^8=0\}\\
&=\{x\in \overline{\GF{2}}\,|\,(x+x^2)(1+x+x^2)(1+x+x^4)=0\}.
\end{align*}
The least field that contains this set is
$\GF{2^4}=\GF{2^{2[2,2]}}$.

That is, generally, $\GF{2^{2[n,k]}}$ is the smallest field
including
\[
\{x\in \overline{\GF{2}}\,|\,T_k(x)\in \GF{2^n}\}.
\]
\end{example}

\begin{proposition}\label{gcd_lcm} When $a\in \GF{2^n}$,
\[
T^{[n,k]}_k(a)=T_d^n(a)=T^{[n,n-k]}_{n-k}(a).
\]
\end{proposition}
\begin{proof}

By definition
$T^{[n,k]}_k(a)=\sum_{i=0}^{\frac{[n,k]}{k}-1}a^{2^{ik}}$,
$T^{n}_d(a)=\sum_{j=0}^{\frac{n}{d}-1}a^{2^{jd}}$ and
$T^{[n,n-k]}_{n-k}(a)=\sum_{i=0}^{\frac{[n,n-k]}{n-k}-1}a^{2^{i(n-k)}}$.
Note that the upper bounds of indices in three summations are
identical: $\frac{[n,k]}{k}=\frac{n}{d}=\frac{[n,n-k]}{n-k}$. It is
easy to check $\{ik \mod n\,|\,0\leq i \leq
\frac{[n,k]}{k}-1\}=\{jd\,|\,0\leq j\leq \frac{n}{d}-1\}=\{i(n-k)
\mod n\,|\,0\leq i \leq \frac{[n,n-k]}{n-k}-1\}$. Since $a^{2^n}=a$,
all three summations are identical. \qed
\end{proof}

\section{Zeros of $T_l^k$}

\begin{lemma}\label{corT_k=0} Followings are
facts.
\begin{enumerate}
\item
\[
\{x\in
\overline{\GF{2}}\,|\,T_k(x)=0\}=T_2(\GF{2^k})=\{x+x^2\,|\,x\in
\GF{2^k}\}.
\]
\item
\[ \{x\in
\GF{2^n}\,|\,T_k(x)=0\}=T_2(\GF{2^k})\cap\GF{2^n}=\{x+x^2\in
\GF{2^n}\,|\,x\in \GF{2^k}\}.
\]
\end{enumerate}
In particular,
\begin{itemize}
\item If $\frac{k}{d}$ is odd, then
\[
\{x\in \GF{2^n}\,|\,T_k(x)=0\}=\{x\in\GF{2^d}\,|\,T_d(x)=0\}.
\]
\item If $\frac{k}{d}$ is even, then
\[
\{x\in \GF{2^n}\,|\,T_k(x)=0\}=\GF{2^d}.
\]

\end{itemize}
\end{lemma}
\begin{proof}
For $x\in \GF{2^k}$,  by Lemma \ref{lem_properties},
$T_k(T_2(x))=x+x^{2^k}=0$, and the two sets $\{x\in
\overline{\GF{2}}\,|\,T_k(x)=0\}$ and $T_2(\GF{2^k})$ have the same
cardinality $2^{k-1}$ and so they coincide. As a immediate
consequence, we have $\{x\in
\GF{2^n}\,|\,T_k(x)=0\}=T_2(\GF{2^k})\cap\GF{2^n}$.

Thus, $\{x\in
\GF{2^n}\,|\,T_k(x)=0\}=T_2(\GF{2^k})\cap\GF{2^n}\subset \GF{2^d}$,
and so $\{x\in \GF{2^n}\,|\,T_k(x)=0\}=\{x\in
\GF{2^d}\,|\,T_k(x)=0\}=\{x\in \GF{2^d}\,|\,\frac{k}{d}T_d(x)=0\}$,
which completes the proof. \qed
\end{proof}

\begin{lemma} \label{ker_T_k_l} Let $l$ be a divisor of $k$. Followings are facts.
\begin{enumerate}
\item
\[ \{x\in \overline{\GF{2}}\,|\,T_l^k(x)=0\}=T_l\circ
T_2(\GF{2^k})=\{x+x^{2^l}\,|\,x\in \GF{2^k}\}.
\]
\item
\[ \{x \in \GF{2^n}\,|\,T_l^k(x)=0\}=T_l\circ
T_2(\GF{2^k})\cap \GF{2^n}=\{x+x^{2^l}\in \GF{2^n}\,|\,x\in
\GF{2^k}\}.
\]
\end{enumerate}

In particular,
\begin{itemize}
\item If $\frac{k}{[d,l]}$ is odd, then
\[
\{x\in
\GF{2^n}\,|\,T_l^k(x)=0\}=\{x\in\GF{2^d}\,|\,T_{(d,l)}^d(x)=0\}=T_{(d,l)}\circ
T_2(\GF{2^d}).
\]
\item If $\frac{k}{[d,l]}$ is even, then
\[
\{x\in \GF{2^n}\,|\,T_l^k(x)=0\}=\GF{2^d}.
\]
\end{itemize}
\end{lemma}
\begin{proof}
Since $T_l^k(T_l\circ T_2(\GF{2^k}))=T_k\circ T_2(\GF{2^k})=0$, the
set $T_l\circ T_2(\GF{2^k})$ with cardinality $2^{k-l}$ is a subset
of $\{x\in \overline{\GF{2}}\,|\,T_l^k(x)=0\}$ with the same
cardinality $2^{k-l}$, i.e., the two sets coincide. Thus, $\{x \in
\GF{2^n}\,|\,T_l^k(x)=0\}=T_l\circ T_2(\GF{2^k})\cap \GF{2^n}\subset
\GF{2^d}$, and we have
\begin{align*}
&\{x\in \GF{2^n}\,|\,T_l^k(x)=0\}=\{x\in
\GF{2^d}\,|\,T_l^k(x)=0\}\\&=\{x\in
\GF{2^d}\,|\,\frac{k}{[d,l]}T_l^{[d,l]}(x)=0\}\\&=\begin{cases}\GF{2^d},
&\text{if $\frac{k}{[d,l]}$ is
even,}\\
\{x\in\GF{2^d}\,|\,T_{(d,l)}^d(x)=0\}=T_{(d,l)}\circ T_2(\GF{2^d}),
&\text{if $\frac{k}{[d,l]}$ is odd,}\end{cases}
\end{align*}
where Proposition \ref{gcd_lcm} was used for the last equality. \qed
\end{proof}

\section{Expression of solutions in closed field}
\begin{proposition} Let $L$ be any positive integer. For any $a\in\GF{2^L}^*$ and $\xi \in  \mu_{2^L+1}\setminus \{1\}$,

\[
\frac{a}{\xi+1}+\GF{2^L}=\{\frac{a}{\xi'+1}\,|\,\xi' \in
\mu_{2^L+1}\setminus \{1\}\}.
\]
\end{proposition}
\begin{proof}
Let $\eta \in \GF{2^L}$. Then we will show
\[\frac{a}{\xi+1}+\eta=\frac{a}{\xi'+1}\] for some $\xi' \in
\mu_{2^L+1}\setminus \{1\}$. Since
$\frac{a}{\xi+1}+\eta=\frac{a}{\frac{a\xi+\eta \xi +\eta}{a+\eta \xi
+\eta}+1}$, it is enough to show \[\xi'=\frac{a\xi+\eta \xi
+\eta}{a+\eta \xi +\eta}\in \mu_{2^L+1}\setminus \{1\}.\] In fact,
obviously $\xi'\neq 1$ and \[\xi'^{2^L}=\frac{a\xi^{2^L}+\eta
\xi^{2^L} +\eta}{a+\eta \xi^{2^L} +\eta}=\frac{a/\xi+\eta/ \xi
+\eta}{a+\eta/ \xi +\eta}=1/\xi'.\] \qed
\end{proof}

\begin{theorem}\label{T_kneq0} Let $a\in \GF{2^n}^*$. Let $L$ be any multiple of the least common multiple $[n,k]$ of
two integers $n$ and $k$. Then, for any $\xi \in
\mu_{2^L+1}\setminus \{1\}$,
\[
x_0=T_k^{L}\circ T_2(\frac{a}{\xi+1})
\]
is a solution to the equation $T_k(x)=a$. In fact, for any $\xi \in
\mu_{2^L+1}\setminus \{1\}$,
\begin{align*}
\{x\in \overline{\GF{2}}\,|\, T_k(x)=a \}&=\{T_k^{L}\circ
T_2(\frac{a}{\zeta+1})\,|\,\zeta \in \mu_{2^L+1}\setminus
\{1\}\}\\
&=T_k^{L}\circ
T_2(\frac{a}{\xi+1}+\GF{2^L}) \\
&=T_k^{L}\circ T_2(\frac{a}{\xi+1})+T_2(\GF{2^k}).
\end{align*}
\end{theorem}
\begin{proof}
Let us set $x=T_k^{L}\circ T_2(\frac{a}{\xi+1})$ for $\xi \in
\mu_{2^L+1}$. Then, by Lemma \ref{lem_properties}, one has
\begin{align*}
&T_k(x)=T_k^{L}\circ T_k \circ T_2(\frac{a}{\xi+1})\\
&=T_L\circ T_2(\frac{a}{\xi+1})\\
&=\frac{a}{\xi+1}+\left(\frac{a}{\xi+1}\right)^{2^L}\\
&=\frac{a}{\xi+1}+\frac{a}{\xi^{2^L}+1}\\
&=\frac{a}{\xi+1}+\frac{a}{1/\xi+1}=a.
\end{align*}
On the other hand, $\#\{x\in \overline{\GF{2}}\,|\, T_k(x)=a
\}=2^{k-1}$, and $\#T_k^{L}\circ
T_2(\frac{a}{\xi+1}+\GF{2^L})=\#T_k^{L}\circ
T_2(\GF{2^L})=T_2(\GF{2^k})=2^{k-1}$. This completes the proof. \qed
\end{proof}

\begin{corollary} Let $a\in \GF{2^n}$ and $l$ be a divisor of $k$. Let $L$ be any multiple of the least common multiple $[n,k]$ of
two integers $n$ and $k$. Then, for any $\xi \in
\mu_{2^L+1}\setminus \{1\}$,
\[
x_0=T_l\circ T_k^{L}\circ T_2(\frac{a}{\xi+1})
\]
is a solution to the equation $T_l^k(x)=a$. In fact, for any $\xi
\in \mu_{2^L+1}\setminus \{1\}$,
\begin{align*}
\{x\in \overline{\GF{2}}\,|\, T_l^k(x)=a \}&=\{T_l\circ T_k^{L}\circ
T_2(\frac{a}{\zeta+1})\,|\,\zeta \in \mu_{2^L+1}\setminus
\{1\}\}\\
&=T_l\circ T_k^{L}\circ
T_2(\frac{a}{\xi+1}+\GF{2^L})\\
&=T_l\circ T_k^{L}\circ T_2(\frac{a}{\xi+1})+T_l\circ T_2(\GF{2^k}).
\end{align*}
\end{corollary}

Note that it is easy to take $\xi \in \mu_{2^{L}+1}\setminus \{1\}$:
Choose any $s\in\GF{2^{2L}}\setminus \GF{2^{L}}$, then calculate
$\xi=s^{2^{L}-1}$.

\section{Solutions in $\GF{2^n}$}

\begin{theorem} For $a\in \GF{2^n}$, the linear equation
\begin{equation}
T_k^{2k}(x)=a   \text{    (i.e. $x^{2^k}+x=a$)}
\end{equation}  has solution in $\GF{2^n}$
if and only if \begin{equation} T_{d}^n(a)=0.
\end{equation} When $T_{d}^n(a)=0$, this equation $T_k^{2k}(x)=a$ has
exactly $2^{d}$ solutions in $\GF{2^n}$:
\begin{itemize}
\item If $\frac{k}{d}$ is odd, then for any $\xi\in \mu_{2^{n}+1}\setminus\{1\}$
\begin{equation}\label{odd_T_k_2} \{x\in
\GF{2^n}\,|\,x^{2^k}+x=a\}=T_k^{[n,k]}(\frac{a}{\xi+1})+\GF{2^d}.
\end{equation}
\item If $\frac{k}{d}$ is even, then for any $\xi\in
\mu_{2^{n}+1}\setminus\{1\}$
\begin{equation}\label{even_T_k_2} \{x\in
\GF{2^n}\,|\,x^{2^k}+x=a\}=T_{n-k}^{[n,n-k]}(\frac{a^{2^{n-k}}}{\xi+1})+\GF{2^d}.
\end{equation}
\end{itemize}
\end{theorem}
\begin{proof}
This easily follows from the fact that the linear operator
$T_k^{2k}(x)=x^{2^k}+x$ on $\GF{2^n}$ has the kernel of dimension
$d$ and, thus, the number of elements in the image of $T_k^{2k}$ is
$2^{n-d}$. For any $x\in \GF{2^n}$, we have
\begin{align*}
T_{d}^n(x^{2^k}+x) &= T_{d}^n(x^{2^k})+T_{d}^n(x)\\
 &=T_{d}^n(x)^{2^k}+T_{d}^n(x)\\
 &=T_{d}^n(x)+T_{d}^n(x) \text{ (since $T_{d}^n(x)\subset \GF{2^{d}}\subset \GF{2^k}$) }\\
 &=0
\end{align*}
leading to the conclusion that the image of $T_k^{2k}$ contains such
all elements in $\GF{2^n}$ since the total number of such elements
in $\GF{2^n}$ is exactly $2^{n-d}$.

Let us prove the second part of the theorem.  First, let us assume
$\frac{k}{d}$ is odd. Then, $x_0=T_k^{[n,k]}(\frac{a}{\xi+1})$ is a
solution to the equation since $T_k\circ
T_2(T_k^{[n,k]}(\frac{a}{\xi+1}))=T_{[n,k]}\circ
T_2(\frac{a}{\xi+1})=\frac{a}{\xi+1}+(\frac{a}{\xi+1})^{2^{[n,k]}}=a$.
On the other hand, under the condition $T^n_d(a)=0$, this solution
really belongs to $\GF{2^n}$. In fact,
$x_0+x_0^{2^n}=T_k^{[n,k]}(\frac{a}{\xi+1}+(\frac{a}{\xi+1})^{2^n})=T_k^{[n,k]}(a)=T^n_d(a)=0$.

If $\frac{k}{d}$ is even, then we will consider a new equation
$x^{2^{n-k}}+x=a^{2^{n-k}}$ instead of the original equation
$x^{2^k}+x=a$. As obvious, this new equation shares the same
$\GF{2^n}$-solution set with the original equation. Since
$\frac{n-k}{d}$ is odd as $\frac{k}{d}$ is even, we can apply the
solution formula \eqref{odd_T_k_2} for odd case to this new
equation. \qed
\end{proof}

\begin{corollary}\label{quadratic_new}
For any $\xi \in \mu_{2^{n}+1}\setminus \{1\}$,
$x_0=T_{n}(\frac{a}{\xi+1})$ and $x_0+1$ are solutions of
$x^2+x+a=0$. These solutions are in $\GF{2^n}$ if and only if
\[
T_n(a)=0.
\]
\end{corollary}

\begin{theorem}\label{Solution_T_k} Let $a\in \GF{2^n}$. Consider the linear equation
\begin{equation}\label{eqT_k}
T_k(x)=a \text{    (i.e. $x+x^2+\cdots+x^{2^{k-1}}=a$)}
\end{equation}
\begin{enumerate}
\item Let $\frac{k}{d}$ be odd. Then,
equation \eqref{eqT_k} has a solution in $\GF{2^n}$ if and only if
\begin{equation}
T_d^n(a)\in \GF{2}.
\end{equation}
 When $T_d^n(a)\in \GF{2}$, the equation
$T_k(x)=a$ has exactly $2^{d-1}$ solutions in $\GF{2^n}$: for any
$\xi\in \mu_{2^{n}+1}\setminus\{1\}$
\begin{equation} \{x\in
\GF{2^n}\,|\,T_k(x)=a\}=T_2\circ
T_k^{[n,k]}(\frac{a}{\xi+1})+T_2(\GF{2^d}).
\end{equation}
\item Let $\frac{k}{d}$ be even. Then,
equation \eqref{eqT_k} has a solution in $\GF{2^n}$ if and only if
\begin{equation}
T_d^n(a)=0.
\end{equation}
 When $T_d^n(a)=0$, the equation $T_k(x)=a$ has
exactly $2^d$ solutions in $\GF{2^n}$: for any $\xi\in
\mu_{2^{n}+1}\setminus\{1\}$
\begin{equation} \{x\in
\GF{2^n}\,|\,T_k(x)=a\}=T_2\circ
T_{n-k}^{[n,n-k]}(\frac{a^{2^{n-k}}}{\xi+1})+\GF{2^d}.
\end{equation}

\end{enumerate}
\end{theorem}

\begin{proof}
Since $T_d^n(\GF{2^n})=\GF{2^d}$ by Corollary \ref{T_image}, one has
\begin{equation}\label{eq}
T_d^n(T_k(\GF{2^n}))=T_k(T_d^n(\GF{2^n}))=T_k(\GF{2^d})=\begin{cases}\GF{2},
&\text{if $\frac{k}{d}$ is odd}\\ 0, &\text{if $\frac{k}{d}$ is
even}.
\end{cases}
\end{equation}
Let us assume $\frac{k}{d}$ is odd. By Corollary \ref{corT_k=0}, we
have $\#T_k(\GF{2^n})=2^{n-d+1}$. Since $\#\{x\in
\overline{\GF{2}}\,|\,T_d^n(x)\in \GF{2}\}=2^{n-d+1}$ as obvious, by
\eqref{eq} we have $\{x\in \overline{\GF{2}}\,|\,T_d^n(x)\in
\GF{2}\}=T_k(\GF{2^n})$, i.e. $T_k(x)=a$ has a solution in
$\GF{2^n}$ if and only if $T_d^n(a)\in \GF{2}$.

At this time, let us assume $\frac{k}{d}$ is even.  Then, by
Corollary \ref{corT_k=0}, $\#T_k(\GF{2^n})=2^{n-d}$.  Since
$\#\{x\in \overline{\GF{2}}\,|\,T_d^n(x)=0\}=2^{n-d}$,  from
\eqref{eq} it follows $\{x\in
\overline{\GF{2}}\,|\,T_d^n(x)=0\}=T_k(\GF{2^n})$, i.e., $T_k(x)=a$
has a solution in $\GF{2^n}$ if and only if $T_d^n(a)=0$.

The assertions about the solution number are consequences of
Corollary \ref{corT_k=0}. The solution formulas were deduced from
\eqref{odd_T_k_2}, \eqref{even_T_k_2} and the fact that if $T_k\circ
T_2(x)=a$, then $y=T_2(x)$ is solution to $T_k(y)=0$. \qed
\end{proof}

\begin{theorem} Let $a\in \GF{2^n}$. Consider the linear equation
\begin{equation}\label{eqT_k_l}
T_l^k(x)=a \text{    (i.e.
$x+x^{2^l}+\cdots+x^{2^{l(\frac{k}{l}-1)}}=a$)}
\end{equation}
\begin{enumerate}
\item Let $\frac{k}{[d,l]}$ be odd. Then,
the equation \eqref{eqT_k_l} has a solution in $\GF{2^n}$ if and
only if
\begin{equation}
T_d^n(a)\in \GF{2^{(d,l)}}.
\end{equation}
 When $T_d^n(a)\in \GF{2^{(d,l)}}$, the equation \eqref{eqT_k_l} has exactly $2^{d-(d,l)}$ solutions in $\GF{2^n}$:
\begin{itemize}
\item If $\frac{k}{d}$ is odd, then for any $\xi\in \mu_{2^{n}+1}\setminus\{1\}$
\begin{equation} \{x\in
\GF{2^n}\,|\,T_l^k(x)=a\}=T_l\circ T_2\circ
T_k^{[n,k]}(\frac{a}{\xi+1})+T_{(d,l)}\circ T_2(\GF{2^d}).
\end{equation}
\item If $\frac{k}{d}$ is even, then for any $\xi\in \mu_{2^{d}+1}\setminus\{1\}$
\begin{equation} \{x\in
\GF{2^n}\,|\,T_l^k(x)=a\}=T_l\circ T_2\circ
T_{n-k}^{[n,n-k]}(\frac{a^{2^{n-k}}}{\xi+1})+T_{(d,l)}\circ
T_2(\frac{T_d^n(a)}{\xi+1})+T_{(d,l)}\circ T_2(\GF{2^d}).
\end{equation}
\end{itemize}
\item Let $\frac{k}{[d,l]}$ be even. Then,
the equation \eqref{eqT_k_l} has a solution in $\GF{2^n}$ if and
only if
\begin{equation}
T_d^n(a)=0.
\end{equation}
 When $T_d^n(a)=0$, the equation \eqref{eqT_k_l} has
exactly $2^d$ solutions in $\GF{2^n}$: for any $\xi\in
\mu_{2^{d}+1}\setminus\{1\}$
\begin{equation} \{x\in
\GF{2^n}\,|\,T_l^k(x)=a\}=T_l\circ T_2\circ
T_{n-k}^{[n,n-k]}(\frac{a^{2^{n-k}}}{\xi+1})+ \GF{2^d}.
\end{equation}

\end{enumerate}
\end{theorem}

\begin{proof}
It holds
\begin{align*}
T_d^n(T_l^k(\GF{2^n}))&=T_l^k(T_d^n(\GF{2^n}))\\&=T_l^k(\GF{2^d})=\frac{k}{[d,l]}T_l^{[d,l]}(\GF{2^d})\\&=\frac{k}{[d,l]}T_{(d,l)}^{d}(\GF{2^d}) \text{  by Proposition \ref{gcd_lcm}}\\
&=\begin{cases} 0, &\text{if $\frac{k}{[d,l]}$ is even,}\\
\GF{2^{(d,l)}}, &\text{if $\frac{k}{[d,l]}$ is odd,}
\end{cases}
\end{align*}
and on the other hand, Corollary \ref{ker_T_k_l} let us know
\begin{align*}
\#T_l^k(\GF{2^n})=\begin{cases} 2^{n-d}, &\text{if $\frac{k}{[d,l]}$ is even,}\\
2^{n-d+(d,l)}, &\text{if $\frac{k}{[d,l]}$ is odd.}
\end{cases}
\end{align*}
Since $T_d^n(x)=a$ has $2^{n-d}$ solutions in the closed field
(indeed in $\GF{2^n}$), thus we conclude
\begin{align*}
T_l^k(\GF{2^n})=\begin{cases} \{a\in\GF{2^n}\,|\,T_d^n(a)=0\}, &\text{if $\frac{k}{[d,l]}$ is even,}\\
\{a\in\GF{2^n}\,|\,T_d^n(a)\in \GF{2^{(d,l)}}\}, &\text{if
$\frac{k}{[d,l]}$ is odd,}
\end{cases}
\end{align*}
which is just the sufficient and necessary condition for existence
of solution in $\GF{2^n}$.

When  $\frac{k}{[d,l]}$ be odd and $\frac{k}{d}$ is even, the
solution formula can be checked as follows: Consider $T_d^n(a)\in
\GF{2^d}$. First, it can be checked that $z_0=T_{(d,l)}\circ
T_2(\frac{T_d^n(a)}{\xi+1})$ is a solution in $\GF{2^d}$ of
$T_l^k(z)=T_d^n(a)$, i.e. $T_{[d,l]}^l(z)=T_d^n(a)$ i.e.
$T_{(d,l)}^d(z)=T_d^n(a)$. For $y_0=T_l\circ T_2\circ
T_{n-k}^{[n,n-k]}(\frac{a^{2^{n-k}}}{\xi+1})\in \GF{2^n}$, it is an
easy exercise to check by direct calculation
$T_l^k(y_0)=T_d^n(a)+a$. So $x_0=y_0+z_0$ is a solution in
$\GF{2^n}$.

In remained cases, the solution formulas are deduced from theorem
\ref{Solution_T_k}, regarding the fact that $T_l(x_0)$ is solution
in $\GF{2^n}$ of $T_l^k(x)=a$ if $x_0\in \GF{2^n}$ is solution of
$T_k(x)=a$. \qed
\end{proof}

As an immediate consequence of these facts, one has (confirms):
\begin{corollary} Followings are true.
\begin{itemize}
\item $T_l^k(x)$ is a 2-to-1 mapping on $\GF{2^n}$ if and only if $d=1$ and $\frac{k}{l}$ is even, or, $d=2$ and both $l$ and $\frac{k}{2l}$ are odd.
\item $T_l^k(x)$ is a permutation on $\GF{2^n}$ if and only if $\frac{k}{l}$ is odd and $d|l$.

Hence, when $\frac{k}{l}$  is odd, $T_l^k(x)$ is an exceptional
polynomial over $\GF{2}$.
\end{itemize}
\end{corollary}


\begin{thebibliography}{10}
\bibitem{BSS1999}
I. Blake, G. Seroussi, N. Smart. Elliptic Curves in Cryptography.
Number 265 in London Mathematical Society Lecture Note Series.
Cambridge University Press, 1999.


\bibitem{MP2013}
G.L. Mullen and D. Panario. Handbook of Finite Fields. Discrete
Mathematics and Its Applications, CRC Press, 2013.
\end{thebibliography}
\end{document}